\documentclass{tlp}

\newtheorem{theorem}{Theorem} 
\newtheorem{definition}[theorem]{Definition} 
\newtheorem{example}[theorem]{Example} 
\newtheorem{lemma}[theorem]{Lemma} 
 
\newtheorem{corollary}[theorem]{Corollary}

\usepackage{amssymb}
\usepackage{amsmath}
\usepackage{stmaryrd}
\usepackage{url}
\usepackage[all]{xy}
\usepackage{wrapfig}
\usepackage{listings}
\usepackage{color}
\let\citet=\cite  %tlp class does not support \citet

\usepackage[breaklinks]{hyperref}

\newcounter{ranking structure implies termination}
\newcounter{underlying domain well founded}
\newcounter{lem:decdiff}
\newcounter{ranking thm}
\newcounter{NP thm}
\newcounter{thm:main}
\newcounter{anchor correct}

\begin{document}
\title[SAT-Based Termination Using 
    Monotonicity Constraints]{
    SAT-Based Termination Analysis Using 
    Monotonicity Constraints over the Integers%
\thanks{Supported by the G.I.F.\ grant 966-116.6.}
}

\author[M.~Codish et al.]{MICHAEL CODISH, IGOR GONOPOLSKIY\\
Department of Computer Science, Ben-Gurion University, Israel
\and AMIR M.\ BEN-AMRAM\\
School of Computer Science, Tel-Aviv Academic College, Israel
\thanks{Part of this author's work was carried out while visiting DIKU, the
University of Copenhagen.}
\and
CARSTEN FUHS, J\"URGEN GIESL\\
LuFG Informatik 2, RWTH Aachen University, Germany
}

%\pagestyle{empty}

% Mike's definitions
\newcommand{\PP}{\mathcal{P}}
\newcommand{\GG}{\mathcal{G}}
\newcommand{\HH}{\mathcal{H}}
\newcommand{\sset}[2]{\left\{~#1  \left|
      \begin{array}{l}#2\end{array}
    \right.     \right\}}

\newcommand{\set}[1]{\left\{ \: #1 \:
%                       \begin{array}{l}#1\end{array}
                     \right\}}
\newcommand\tuple[1]{\langle #1 \rangle}
\newenvironment{SProg}
     {\begin{small}\begin{tt}\begin{tabular}[c]{l}}%
     {\end{tabular}\end{tt}\end{small}}
\newcommand{\qin}{\hspace*{0.15in}}
\newcommand{\scgarrow}{\mathchoice{\mbox{\rm ~:--~}}{\mbox{\rm ~:--~}}{\mbox{\rm :--}\;}{\mbox{\rm :--}\;}    }

% Amir's definitions
\newcommand{\julia}{\textsf{Julia}}
\newcommand{\val}{\mbox{\it Val}}
% a symbol for "down or equal" symbol
\newcommand{\deq}{{\downarrow}{\raisebox{1ex}{\rmfamily\hspace{-1.3ex}\scriptsize=}}}
\newcommand{\deqsm}{\downarrow{}^{\hspace{-1.1ex}=}} %a subscript version
\newcommand{\bdfn}{\begin{definition}}
\newcommand{\edfn}{\end{definition}}
\newcommand{\eqdef}{\triangleq}

\newcommand{\nats}{\mathbb N}
\newcommand{\extInts}{{\overline{\mathbb Z}}}

% Carsten's definitions
\newcommand{\aprove}{\textsf{AProVE}}
\newcommand{\costa}{\textsf{COSTA}}
\newcommand{\FFF}{\mathsf{F}}
\newcommand{\Ftrue}{\mathsf{true}}
\newcommand{\pol}[1]{||#1||}
\newcommand{\Fp}{\mathsf{p}}
\newcommand{\Fq}{\mathsf{q}}

\newcommand{\ints}{\mathbb Z}
\newcommand{\rationals}{\mathbb Q}

\maketitle  
\label{firstpage}

\begin{abstract}
  We describe an algorithm for proving termination of programs
  abstracted to systems of monotonicity constraints in the integer
  domain.  Monotonicity constraints are a non-trivial extension of the
  well-known size-change termination method.  While deciding
  termination for systems of monotonicity constraints is PSPACE
  complete, we focus on a well-defined and significant subset, which
  we call MCNP, designed to be amenable to a SAT-based solution.  Our
  technique is based on the search for a special type of ranking
  function defined in terms of bounded differences between multisets
  of integer values.  We describe the application of our approach as
  the back-end for the termination analysis of Java Bytecode (JBC). At
  the front-end, systems of monotonicity constraints are obtained by
  abstracting information, using two different termination analyzers:
  \aprove\ and \costa.  Preliminary results reveal that our approach
  provides a good trade-off between precision and cost of analysis.
\end{abstract}


\section{Introduction} 

Proving termination is a fundamental problem in verification.  The
challenge of termination analysis is to design a program abstraction
that captures the properties needed to prove termination as often as
possible, while providing a decidable sufficient criterion for
termination.  
Typically, such abstractions represent a program as a finite set of
abstract transition rules which are descriptions of program steps, where
the notion of step can be tuned to different needs. The abstraction
considered in this paper is based on monotonicity-constraint systems
(MCSs).

The MCS abstraction is an extension of the SCT (size-change
termination~\cite{leejonesbenamram01}) abstraction, which has been
studied extensively during the last decade ({see
  \url{http://www2.mta.ac.il/~amirben/sct.html} for a summary and
references}).
In the SCT abstraction, an abstract transition rule is specified by a set
of inequalities that show how the sizes of program data in the target
state are bounded by those in the source state. Size is measured by a
well-founded base order. These inequalities are often represented by a
\emph{size-change graph}.

The size-change technique was conceived to deal with well-founded
domains, where infinite descent is impossible.  Termination is deduced
by proving that any (hypothetical) infinite run would decrease some
value monotonically and endlessly, so that well-foundedness would be
contradicted.

Extending this approach, a \emph{monotonicity constraint} (MC) allows
for any conjunction of order relations (strict and non-strict
inequalities) involving any pair of variables from the source
and target states. So in contrast to SCT, one may also have
relations between two variables in the target state or two variables
in the source state.
Thus, MCSs are more expressive, and \cite{Codish-et-al:05} 
observe that earlier analyzers based on monotonicity
constraints \cite{LS:97,CT:99,LSS:04} apply a termination test which
is sound and complete for SCT, but incomplete for monotonicity
constraints, even if one does  
 \begin{wrapfigure}[9]{r}{0.35\textwidth}
 \vspace*{-.4cm}
  \hspace*{-.2cm}\begin{minipage}[t]{\linewidth}
       \scriptsize
       \lstset{language=Java,basicstyle=\tt}
       \fbox{\parbox{4.35cm}{\vspace*{-.25cm}
\begin{lstlisting}
static int a(int x, int y){
  if (x>y){
    int x1=x-1; int y1=y+1;
    if (x1>=y1) 
      return a(x1,y1); 
    else return y;
  } else {
    int x1=x+1; int y1=y-1;
    if (x1<=y1) 
      return a(x1,y1); 
    else return x;
  }
}
\end{lstlisting}\vspace*{-.28cm}}}
   \end{minipage}
 \end{wrapfigure}
 not change the underlying model, namely that ``data'' are from an
 unspecified well-founded domain.
They also point out that monotonicity constraints
can imply termination under a different assumption---that the data are
integers. Not being well-founded, integer data cannot be handled by SCT.
As an example, consider the Java program on the right which computes
the average of \texttt{x} and \texttt{y}. The loops in this program
can be abstracted to the following monotonicity-constraint transition
rules:
\[
  \begin{array}{llcll}
  (1) \qquad & a(x,y) & \scgarrow & x>y,    x>x', y'>y, x'\geq y'; &  a(x',y')\\
  (2) & a(x,y) & \scgarrow & y \geq x,x'>x, y>y', y'\geq x'; &  a(x',y')
\end{array}
\]

To prove termination of the Java program it is sufficient to focus on
the corresponding abstraction.  Note that termination of this program
cannot be proved using SCT, not only because SCT disallows
constraints between source variables (such as $x{>}y$), but also because
it computes with integers rather than  natural numbers.

To see how the transition constraints imply termination, observe
that if (1) is repeatedly taken, then the value of $y$ grows;
constraint $x>y$ (with the fact that $x$ descends) implies
that this cannot go on forever.
In (2), the situation is reversed: $y$ descends and is lower-bounded
by $x$. In addition, constraint $y'\ge x'$ of rule (2) implies that,
once this rule is taken, there can be no more applications of
(1). Therefore any (hypothetical) infinite computation would
eventually enter a loop of (1)s or a loop of (2)s;
possibilities which we have just ruled out.  In this paper, we show
how to obtain such termination proofs automatically using SAT solving.

Although MCS and SCT are abstractions where termination is
decidable, they have a drawback: the decision problems are PSPACE
complete and a certificate for termination under these abstractions
can be of prohibitive complexity (not ``polynomially computable''
\cite{Ben-Amram:ranking}).  Typical implementations based on the SCT 
abstraction apply a closure operation on transition rules
which is exponential both in time and in space.
\cite{TACAS08} addressed this problem for SCT,
identifying an NP complete subclass of SCT, called SCNP, which yields
polynomial-size certificates. Moreover, \cite{TACAS08} automated SCNP
using a SAT solver.  Experiments indicated that, in practice, this
method had good performance and power when compared to a complete SCT
decision procedure, and had the additional merit of producing
certificates.

In this paper we tackle the similar problem to prove termination of
monotonicity-constraint systems in the integer domain.  As noted
above, the integer setting is more complicated than the well-founded
setting.  Termination is often proved by looking at
\emph{differences} of certain program values (which should be
decreasing and lower-bounded).  One could simulate such
reasoning in SCT by creating fresh variables to track the non-negative
differences of pairs of original variables. However this loses precision 
and may square the number of variables, which
is an exponent in the complexity of most SCT algorithms.  Instead, we
use an idea from \cite{TACAS08} which consists of mapping program
states into multisets of argument values.  The adaption of this method
to integer data is non-trivial.
Our new solution uses the following ideas: (1) We associate two sets
with each program point and define how to ``subtract'' them so that
the difference can be used for ranking (generalizing the difference of
two integers). This avoids the quadratic growth in the exponent of the
complexity, since we are only working with the original variables and
relations, and is also more expressive.  (2) We introduce a concept of
``ranking functions'' which is less strict than typically used but
still suffices for termination. It allows the co-domain
of the function to be a non-well-founded set that has a well-founded
subset.  This gives an additional edge over the na\"ive reduction to
SCT, which can only make use of differences which are definitely
non-negative.

After presenting preliminaries in Sect.\ \ref{Monotonicity-Constraint
  Systems and Their Termination}, Sect.\ \ref{Ranking Functions for
  Monotonicity Constraint Systems} introduces 
\emph{ranking structures}, which are termination witnesses.  In Sect.\
\ref{sec-sat} we show that such a witness can be verified in
polynomial time, hence the resulting subclass of terminating MCSs lies
in NP. Consequently, we call it MCNP. In Sect.\
\ref{subsec:algorithm} we devise an algorithm that uses a SAT solver
as a back-end to solve the resulting search problems.
Sect.\ \ref{sec:exp} describes an empirical evaluation using a
prototypical implementation as the back-end for termination analysis
of Java Bytecode (JBC).  Results indicate a good trade-off between
precision and cost of analysis.  All proofs and further details of the
evaluation can be found in the appendices.

\emph{Related work.} 
Termination analysis is a vast field and we focus here on the most
closely related work.  On termination analyzers for JBC, we mention
\costa\ \cite{FMOODS08}, \julia~\cite{Julia-TOPLAS}, and \aprove\
\cite{JBC-Correctness,RTA10}.  Both \costa\ and \julia\ abstract
programs into a CLP form, as in this work; but use a richer constraint
language that makes termination of the abstract program
undecidable. On extending SCT to the integer domain: \cite{Avery:06}
uses constraints of the form ${x{>}y'},{x{\ge} y'},{x{<}y'},{x{\le}
  y'}$ along with polyhedral state invariants (similar constraints as
those used by \costa\ and \julia) to find lower-bounded combinations
of the variables.  \cite{MV-cav06} uses SCT constraints on
pseudo-variables that represent ``measures'' invented by the
system. This allows it to handle integers by taking, for example, the
differences of two variables as a measure.  \cite{DLSS:2001,SdS:2004}
prove termination of logic programs that depend on numerical
constraints by inferring ``level mappings'' based on constraints
selected from the source program; so, a constraint like $x>y$ can
trigger the use of $x-y$ as a level mapping.  There are numerous
applications of SAT for deciding termination problems for all kinds of
programs (e.g., one of the first such papers is \cite{RTA2006}).

\section{Monotonicity-Constraint Systems and Their Termination}
\label{Monotonicity-Constraint Systems and Their Termination}

Our method is programming-language independent. It works on an
abstraction of the program provided by a front-end.  An
abstract program is a transition system with states expressed in terms
of a finite number of variables (\emph{argument positions}).

\begin{definition}[constraint transition system] \label{def:abstract program} 
  A \emph{constraint transition system} is an abstract program,
  represented by a directed multigraph called a \emph{control-flow
    graph} (CFG). The vertices are called \emph{program points} and
  they are associated with fixed numbers (arity) of \emph{argument
    positions}. We write $p/n$ to specify the arity of vertex $p$.
  A \emph{program state} is an association of a value from the
  \emph{value domain} to each argument position of a program point
  $p$, denoted $p(x_1,\dots,x_n)$ and abbreviated $p(\bar x)$.  The
  set of all states is denoted $\mathit St$.
  The arcs of the CFG are associated with \emph{transition rules},
  specifying relations on program states, which we write as $p(\bar x)
  \scgarrow \pi;\, q(\bar y)$. The \emph{transition predicate} $\pi$
  is a formula in the \emph{constraint language} of the abstraction.
\end{definition}

Note that a state corresponds to a ground atom: argument positions are
associated with specific values. In a transition rule, positions are
associated with variables that can only be constrained through
$\pi$. Thus in the notation $p(\bar x)$, $\bar x$ may represent ground
values or variables, according to context.
The constraint language in our work is that of \emph{monotonicity
  constraints}.

\begin{definition}[monotonicity constraint] \label{def-scg} A
  \emph{monotonicity constraint} (MC) $\pi$ on $V = \bar x\cup\bar y$
  is a conjunction of constraints $x \rhd y$ where $x,y\in V$, and
  ${\rhd}\in\{>,\ge\}$.  We write $\pi \models x\rhd y$ whenever
  $x\rhd y$ is a consequence of $\pi$ (in the theory of total
  orders). This consequence relation is easily computed, e.g., by a
  graph algorithm.
  A transition rule $p(\bar x) \scgarrow \pi; \, q(\bar y)$, where
  $\pi$ is a MC, is also known as a \emph{monotonicity-constraint
    transition rule}.
  An \emph{integer monotonicity-constraint transition system} (MCS)%
  \footnote{In this work only the integer domain is of interest, hence
    ``integer'' will be omitted.}  is a constraint transition system
  where the value domain is $\ints$ and transition predicates are
  monotonicity constraints.
 \end{definition}

 It is useful to represent a MC as a directed graph (often denoted by
 the letter $g$), with vertices $\bar x\cup\bar y$, and two types of
 edges $(x,y)$: weak and strict. If $\pi \models x > y$ then there is
 a strict edge from $x$ to $y$ and if $\pi\models x \geq y$ (but not
 $x>y$) then the edge is weak.
 Note that there are two kinds of graphs, those representing
 transition rules and the CFG.  We often identify an abstract program
 with its set $\GG$ of transition rules, the CFG being implicitly
 specified.

\begin{definition}[run, termination]
  Let $\GG$ be a transition system.  A \emph{run} of $\GG$ is a
  sequence $p_0(\bar x_0) \stackrel{\pi_0}{\to} p_1(\bar x_1)
  \stackrel{\pi_1}{\to} p_2(\bar x_2) \dots$ of states labeled by
  constraints such that each labeled pair of states, $p_i(\bar x_i)
  \stackrel{\pi_i}{\to} p_{i+1}(\bar x_{i+1})$, corresponds to a
  transition rule $p_i(\bar x) \scgarrow \pi_i; \, p_{i+1}(\bar y)$ from
  $\GG$ (identical except that variables $\bar x$ and $\bar y$ are
  replaced by values $\bar x_i$ and $\bar x_{i+1}$) and such that
  $\pi_i$ is satisfied.
  A transition system \emph{terminates} if it has no infinite run.
  
\end{definition}

\vspace{-3mm}
\begin{example} \label{ex:mainexample}
 This example presents a MCS in
  textual form as well as graphical form.  This system is terminating,
  and in the following sections we shall illustrate how our method
  proves it.  In the graphs, solid arrows stand for strict
  inequalities and dotted arrows stand for weak inequalities.

\medskip\noindent
\mbox{\small 
$\begin{array}{ll@{\;}l@{\;\,}l}
g_1= & p(x_1,x_2,x_3) &\scgarrow~ y_1>x_1, y_2\geq x_1, 
       x_2 \geq y_2, x_2\geq y_3, x_2\geq x_1; &p(y_1,y_2,y_3) \\
g_2= & p(x_1,x_2,x_3) &\scgarrow~ y_1\geq x_1, y_1>x_2, y_2>x_2, x_3\geq y_2,
        x_3\geq y_3, x_3 > x_2; & p(y_1, y_2, y_3) \\
g_3= & p(x_1,x_2,x_3) &\scgarrow~
     y_1>x_1,\, x_2\geq y_2;
 &q(y_1,y_2)\\
g_4= & q(x_1,x_2) &\scgarrow~
 y_1\geq x_1,\, x_2\geq y_2, x_2\geq y_3,\,x_2\geq x_1; & p(y_1,y_2,y_3) 
\end{array}$}

\smallskip

\noindent
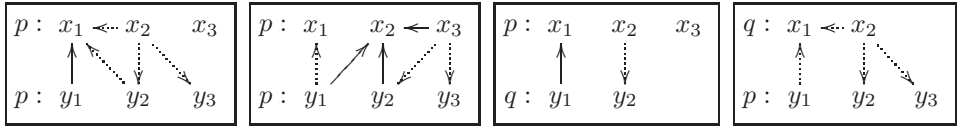
\fbox{\hspace*{-.1cm}\xymatrix@R=15pt{p:\\ p: }
  \xymatrix@C=9pt@R=15pt{
     x_1 & x_2\ar@{.>}[l]\ar@{.>}[dr]\ar@{.>}[d] & x_3\\
     y_1\ar[u] & y_2\ar@{.>}[ul] & y_3
}}\hspace*{.1cm}
\fbox{\hspace*{-.1cm}\xymatrix@R=15pt{ p:\\ p: }
\xymatrix@C=9pt@R=15pt{
     x_1 & x_2 & x_3\ar[l]\ar@{.>}[dl]\ar@{.>}[d]  \\
     y_1\ar@{.>}[u]\ar[ur] & y_2\ar[u] & y_3
}}\hspace*{.1cm}
\fbox{\hspace*{-.1cm}\xymatrix@R=15pt{ p:\\ q: }
  \xymatrix@C=8pt@R=15pt{
     x_1 & x_2\ar@{.>}[d] & x_3 \\
     y_1\ar[u] & y_2
}}\hspace*{.1cm}
\fbox{\hspace*{-.1cm}\xymatrix@R=15pt{ q:\\ p: }
\xymatrix@C=8pt@R=15pt{
     x_1 & x_2 \ar@{.>}[l]\ar@{.>}[dr]\ar@{.>}[d]  \\
     y_1\ar@{.>}[u] & y_2 & y_3
}}
\end{example}

\section{Ranking Structures for Monotonicity-Constraint Systems}
\label{Ranking Functions for Monotonicity Constraint Systems}

This section describes \emph{ranking structures}, a concept that we
introduce for proving termination of MCSs.  Sect.\ \ref{Some
  preliminaries regarding ranking functions} presents the necessary
notions in general form.  Then, Sect.\ \ref{Multiset-based ranking
  functions} specializes them to the form we use for MCNP.

\vspace{-5mm}
\subsection{Ranking structures}
\label{Some preliminaries regarding ranking functions}

Recall that $\succsim$ is a \emph{quasi-order} if it is transitive and
reflexive; its \emph{strict part} $x \succ y$ is the relation $(x
\succsim y) \land (y \not\succsim x)$; the quasi-order is
\emph{well-founded} if there is no infinite chain with $\succ$.  A set
is well-founded if it has a tacitly-understood well-founded order.

A \emph{ranking function} maps program states into a well-founded set,
such that every transition decreases the function's value.  As shown
in~\cite{BA:mcInts}, for every terminating MCS there exists a
corresponding ranking function. However, these are of exponential size
in the worst case. Since our aim is NP complexity, we cannot use that
construction, but instead restrict ourselves to polynomially sized
termination witnesses.  These witnesses, called \emph{ranking
  structures}, are more flexible than ranking functions, and suffice
for most practical termination proofs.

\begin{definition}[anchor, intermittent ranking function] 
  \label{def:anchor} Let $\GG$ be a MCS with state space $St$.  Let
  $(\mathcal D,\succsim)$ be a quasi-order and ${\mathcal D}_+$ a
  well-founded subset of $\cal D$.  Consider a function
  $\Phi:St\to{\mathcal D}$.  We say that $g \in \GG$ is a $\Phi$-\emph{anchor}
  for $\GG$ (or that $g$ is \emph{anchored} by $\Phi$ for $\GG$) if
  for every run $p_0(\bar x_0)\stackrel{\pi_0}{\to} p_1(\bar
  x_1)\stackrel{\pi_1}{\to} \ldots \stackrel{\pi_{k-1}}{\to} p_k(\bar
  x_k) \stackrel{\pi_{k}}{\to} p_{k+1}(\bar x_{k+1}) $ where both
  $p_0(\bar x_0)\stackrel{\pi_0}{\to} p_1(\bar x_1)$ and $p_k(\bar
  x_k) \stackrel{\pi_{k}}{\to} p_{k+1}(\bar x_{k+1})$ correspond to
  the transition rule $g$, we have $\Phi(p_{i}(\bar
  x_{i}))\succsim\Phi(p_{i+1}(\bar x_{i+1}))$ for all $0\leq i\leq k$,
  where at least one of these inequalities is strict; and
  $\Phi(p_{i}(\bar x_{i}))\in {\mathcal D}_+$ for some $0 \leq i \leq
  k$.
A function $\Phi$ which satisfies the above conditions is called 
an \emph{intermittent ranking function} (IRF).\footnotemark
\footnotetext{The term ``intermittent ranking function'' is
              inspired by \cite{MannaWaldinger78}.}
\end{definition}

\vspace{-2mm}
\begin{example} \label{ex:ranking} Consider the transition rules from
  Ex.\ \ref{ex:mainexample}. Let ${\GG} = \{g_1, g_2\}$ and let
  $\Phi_1(p(\bar x)) = max(x_2,x_3)-x_1$.  In any run built with $g_1$
  and $g_2$, the value of $\Phi_1$ is non-negative at least in every
  state followed by a transition by $g_1$. Moreover, a transition by $g_1$
  decreases the value strictly and a transition by $g_2$
  decreases it weakly. Hence, $g_1$ is anchored by $\Phi_1$ for $\GG$
  (in Sect.~\ref{subsec:multisets}, we come back to this example and
  show how $\Phi_1$ fits the patterns of termination proofs that our
  method is designed to discover).
\end{example}

\vspace{-2mm}
\begin{definition}[ranking structure] \label{def:structure} Consider
  $\GG$ and $\mathcal D$ as in Def.\
  \ref{def:anchor}.
  Let $\Phi_1, \dots, \Phi_m:St\to{\mathcal D}$.  Let $\GG_1$
  consist of all transition rules $g \in \GG$ where $\Phi_1$ anchors
  $g$ for $\GG$.  For $2 \leq i \leq m$, let $\GG_i$ consist of all
  transition rules $g \in \GG \setminus (\GG_1 \cup \ldots \cup
  \GG_{i-1})$ where $\Phi_i$ anchors $g$ in $\GG \setminus (\GG_1 \cup
  \ldots \cup \GG_{i-1})$. We say that $ \tuple{\Phi_1,  \dots,
    \Phi_m }$ is a \emph{ranking structure} for $\GG$ if $\GG_1 \cup
  \ldots \cup \GG_m = \GG$.
\end{definition}
Note that by the above definition, for every $g \in \GG$ there is a
(unique) $\GG_i$ with $g \in \GG_i$. We denote this index $i$ as
$i(g)$ (i.e., $g \in \GG_{i(g)}$ for all $g \in \GG$).

\vspace{-1mm}
\begin{example}
  For the program $\{g_1,g_2\}$ of Ex.~\ref{ex:mainexample}, a ranking
  structure is $\tuple{\Phi_1,\Phi_2}$ with $\Phi_1$ as in
  Ex.~\ref{ex:ranking} and $\Phi_2(p(\bar x)) = x_3-x_2$. Here, we
  have $i(g_1)=1$ and $i(g_2)=2$. Later, in Ex.~\ref{ex:multiset-maps}
  and \ref{MCNP example} we will extend the ranking structure to the
  whole program $\{g_1,g_2,g_3,g_4\}$.
\end{example}

The concept of ranking structures generalizes that of lexicographic
global ranking functions used, e.g., in~\cite{TACAS08,ADFG:2010}.  A
lexicographic ranking function is a ranking structure, however, the
converse is not always true, since the function $\Phi$ does not
necessarily decrease on a transition rule which it anchors, and
because $\Phi$ may assume values out of ${\mathcal D}_+$ in certain
states.

\setcounter{ranking structure implies termination}{\value{theorem}}
\addtocounter{ranking structure implies termination}{1}

\begin{theorem} 
If there is a ranking structure for $\GG$, then $\GG$ terminates.
\end{theorem}

\begin{definition}\label{def:irredundant}
  A ranking structure $\tuple{ \Phi_1, \Phi_2, \dots, \Phi_m }$ for
  $\GG$ is \emph{irredundant} if for all $j\le m$, there is a
  transition $g\in\GG$ such that $i(g) = j$.
\end{definition}

\noindent
It follows easily from the definitions that if there is a ranking
structure for $\GG$, there is an irredundant one, of length at most
$|{\cal G}|$.

\subsection{Multiset Orders and Level Mappings}
\label{Multiset-based ranking functions}
\label{subsec:multisets}

The building blocks  for our construction are four quasi-orders on
multisets of integers, and a notion of \emph{level mappings}, which
map program states into pairs of multisets, whose \emph{difference}
(not set-theoretic difference; 
%\pagebreak 
see Def.~\ref{def:ms-diff} below) will
be used to rank the states.\footnote{ A reader familiar with previous
  works using this term should note that here, a level mapping is not
  in itself some kind of ranking function.}  The difference will be
itself a multiset, and we now elaborate on the relations that we use
to order such multisets.

\vspace{-1mm}
\begin{definition}[multiset types]
\label{def:ms-exts}
Let $\wp_n(\ints)$ denote the set of multisets of  integers of at most
$n$ elements, where $n$ is fixed by context.\footnote{For
monotonicity-constraint systems, $n$ is the maximum arity of program points.}
The $\mu$-ordered multiset type, for $\mu\in \set{max, min, ms, dms}$,
is the quasi-ordered set $(\wp_n(\ints),\succsim^\mu)$ where:
\begin{enumerate}
\item \emph{(max order)} $S\succsim^{max} T$ holds iff 
  $max(S) \ge max(T)$, or $T$ is empty;
  $S\succ^{max} T$ holds iff 
  $max(S) > max(T)$, or $T$ is empty while $S$ is not.
\item \emph{(min order)} $S\succsim^{min} T$ holds iff 
  $min(S) \ge min(T)$, or $S$ is empty;
  $S\succ^{min} T$ holds iff
  $min(S) > min(T)$, or $S$ is empty while $T$ is not.
\item \emph{(multiset order~\cite{DeMa:79})} $S\succ^{ms}T$
  holds iff
$T$ is obtained by replacing a non-empty 
  $U\subseteq S$ by a (possibly empty) multiset $V$ such that
  $U\succ^{max} V$;
 the weak relation
  $S\succsim^{ms}T$ holds iff $S\succ^{ms}T$ or $S = T$.
\item \emph{(dual multiset order~\cite{BL:2006})} $S\succ^{dms}T$
holds iff
$T$ is obtained by replacing a sub-multiset
  $U\subseteq S$ by a non-empty multiset $V$ with $U\succ^{min} V$;
  the weak relation
  $S\succsim^{dms}T$ holds 
  iff $S\succ^{dms}T$ or $S = T$.
\end{enumerate}
\end{definition}

\vspace{-3mm}
\begin{example}   
  For $S=\{10,8,5\},\ T=\{9,5\}$:\hspace{0.75em} $S \succ^{max}
  T$,\hspace{0.75em} $T \succsim^{min} S$,\hspace{0.75em} $S
  \succ^{ms}T$, and $T \succ^{dms} S$.
\end{example}

\vspace{-3mm}
\begin{definition}[well-founded subset of multiset types]
  For $\mu\in \set{max, min, ms, dms}$, we define
  $(\wp_n(\ints),\succsim^\mu)_+$ as follows: For \textit{min}
  (respectively \textit{max}) order, the subset consists of the
  multisets whose minimum (resp.~maximum) is non-negative.  For
  \textit{ms} and \textit{dms} orders, the subset consists of the
  multisets all of whose elements are non-negative. 
\end{definition}

\setcounter{underlying domain well founded}{\value{theorem}}
\addtocounter{underlying domain well founded}{1}

\vspace{-3mm}
\begin{lemma}
  For all $\mu\in \{max, min, ms, dms\}$,
  $(\wp_n(\ints),\succsim^\mu)$ is a total quasi-order, with
  $\succ^\mu$ its strict part; and
  \mbox{$(\wp_n(\ints),\succsim^\mu)_+$} is well-founded. 
\end{lemma}

For MCs over the integers, it is necessary to consider differences: in
the simplest case, we have a ``low variable'' $x$ that is
non-descending and a ``high variable'' $y$ that is non-ascending, so
$y-x$ is non-ascending (and will decrease if $x$ or $y$ changes). If
we also have a constraint like $y\ge x$, to bound the difference from
below, we can use this
 \begin{wrapfigure}[4]{r}{19mm}\vspace{-6mm}
 \begin{center}\scriptsize  
 $\fbox{
   \xymatrix@C=15pt@R=15pt{
      x   & y\ar@{.>}[l]\ar@{.>}[d] \\
      x'\ar@{.>}[u] & y'
 }}$
   \end{center}
 \end{wrapfigure}
  for ranking a loop
(we refer to this situation as ``the $\mathrm\Pi$''---due
to the diagram on the right).
In the more general case, we consider sets of variables. We will
search for a similar $\mathrm \Pi$ situation involving a ``low set''
and a ``high set''.  We next define how to form a difference of two
sets so that one can follow the same strategy of ``diminishing
difference''.

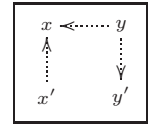
\begin{definition}[multiset difference]
\label{def:ms-diff}
Let $L,H$ be non-empty multisets with types $\mu_L, \mu_H$ respectively.  Their
difference $H-L$ is defined in the following way, depending on the
types 
%\pagebreak 
(there are 6 cases):
 \begin{enumerate}
 \item For $\mu_L \in\{max,min\}$, $H-L = \{ h- \mu_L(L)~|~h\in H \}$
   and has the type of $H$. \\ (Here, $\mu_L(L)$ signifies $min(L)$ or
   $max(L)$ depending on the value of $\mu_L$).
 \item For $\mu_L\in\{ms,dms\}$ and $\mu_H\in\{min,max\}$,   
   $H-L = \{ \mu_H(H) - \ell~|~\ell\in L\}$ and has type $\overline\mu_L$
   (where $\overline{ms}=dms$ and $\overline{dms}=ms$).
\end{enumerate}
For $L$ and $H$ such that $H-L$ is defined, we say that the types of
$L$ and $H$ are \emph{compatible}. We write $H\supsetpluseq L$ if the
difference belongs to the well-founded subset.
\end{definition}

Note that $\supsetpluseq$ relates multisets of possibly different
types and is not an order relation. 
Termination proofs do not require to define the difference of
multisets with types in $\{ms,dms\}$.  To see why, observe that in
``the $\mathrm \Pi$'', only one multiset must change strictly, and the
non-strict relations $\succsim^{ms}$, $\succsim^{dms}$ are contained
in $\succsim^{max}$, $\succsim^{min}$, respectively.
Note also that $H\supsetpluseq L$ is equivalent, in all relevant
cases, to $\mu_1(H) \ge \mu_2(L)$ with $\mu_1, \mu_2 \in \{min,max\}$.
The intuition into why multiset difference is defined as above is
rooted in the following lemma.

\setcounter{lem:decdiff}{\value{theorem}}
\addtocounter{lem:decdiff}{1}

\begin{lemma} 
  Let $L, H$ be two multisets of compatible types $\mu_L, \mu_H$, and
  let $\mu_D$ be the type of $H-L$.  Let $L', H'$ be of the same types
  as $L, H$ respectively.  Then

\vspace*{-.2cm}

{\small\begin{eqnarray*}
H\succsim^{\mu_H} H' \land L \precsim^{\mu_L} L'   
             &\Longrightarrow& H-L \succsim^{\mu_D} H'-L'  ;\\
H\succ^{\mu_H} H' \land L \precsim^{\mu_L} L'
             & \Longrightarrow&  H-L \succ^{\mu_D} H'-L'  ;\\
H\succsim^{\mu_H} H' \land L \prec^{\mu_L} L'
             & \Longrightarrow&  H-L \succ^{\mu_D} H'-L'  \,.
\end{eqnarray*}}
\end{lemma}

\noindent
\emph{Level mappings} are functions that facilitate the construction
of ranking structures.
Three types of level mappings are defined in~\cite{TACAS08}:
\emph{numeric}, \emph{plain}, and \emph{tagged}.  In this paper we
focus on ``plain'' and ``tagged'' level mappings and we adapt them for
multisets of integers. Numeric level mappings have become redundant in
this paper due to the passage from ranking functions to ranking
structures.  We first introduce the extension for plain level
mappings.
\begin{definition}[bi-multiset level mapping, or ``level mapping'' for short]
\label{def:bi-multiset level mapping}
Let $\GG$ be a MCS.  A \emph{bi-multiset level mapping},
$f_{\mu_L,\mu_H}$ maps each program state $p(\bar x)$ to a pair of
(possibly intersecting) multisets $p_f^{low}(\bar x) =
\set{u_1,\ldots,u_l}\subseteq\bar x$ and $p_f^{high}(\bar x) =
\set{v_1,\ldots,v_k}\subseteq\bar x$ with types indicated respectively
by $\mu_L,\mu_H\in\set{max,min,ms,dms}$. Only compatible pairs
$\mu_L,\mu_H$ are admitted.
  The selection of argument positions only depends on the program
  point $p$.
\end{definition}

\vspace{-1mm}
\begin{example} 
\label{ex:multiset-maps}
The following are the level mappings  
used (in Ex.\ \ref{MCNP example})
to prove termination of the program of Ex.\ \ref{ex:mainexample}.
Here, each program point $p$ is mapped to $\tuple{p_{f}^{low}(\bar x),
  p_{f}^{high}(\bar x)}$.

\begin{minipage}[t]{0.5\linewidth}
  $\begin{array}{l}
  f^1_{min,max}(p(\bar x)) = \tuple{\set{x_1},\set{x_2,x_3}} \\
  f^1_{min,max}(q(\bar x)) = \tuple{\set{x_1},\set{x_2}} \\
  \end{array}$
\end{minipage}
\begin{minipage}[t]{0.5\linewidth}
$\begin{array}{l}
f^2_{min,max}(p(\bar x)) = \tuple{\set{x_2},\set{x_3}}\\
f^2_{min,max}(q(\bar x)) = \tuple{\set{},\set{}} 
\end{array}$
\end{minipage}
\end{example}

We now turn to tagged level mappings. Assume the context of
Def.~\ref{def:bi-multiset level mapping} and let $M$ denote the sum of
the arities of all program points.
A \emph{tagged} bi-multiset level mapping is just like a bi-multiset
level mapping, except that set elements are pairs of
the form $(x,t)$ where $x$ is from $\bar x$ and $t<M$ is a natural
constant, called a tag.
We view such a pair as representing the
integer value $Mx+t$ (recall that $x$ is an integer).
This transforms
tagged multisets into multisets of 
%\pagebreak 
integers, so
Defs.~\ref{def:ms-diff}, \ref{def:bi-multiset level mapping}, and the
consequent definitions and results can be used without  change.

Tags ``prioritize'' certain argument positions and can usefully turn
weak inequalities into strict ones.  For example, consider a
transition rule $ p(\bar x) \scgarrow x_1 > y_1, x_1\ge y_2,\dots;
p(\bar y)$. The tagged set $\{(x_1,1), (x_2,0)\}$ is strictly greater
(in $ms$ order as well as in $max$ order) than $\{(y_1,1), (y_2,0)\}$
(because $\pi\models (x_1,1)>(y_2,0)$).  The plain sets $\{x_1,x_2\}$
and $\{y_1,y_2\}$ do not satisfy these relations. Thus tagging may
increase the chance of finding a termination proof.  We do not have
any fixed rule for tagging; our SAT-based procedure will find a useful
tagging if one exists.  
In the remainder we write ``level mapping'' to indicate a, possibly
tagged, bi-multiset level mapping.

Level mappings are applied in termination proofs to
express the diminishing difference of their low and high sets. To be
useful, we also need to express a constraint relating the high and low
sets, providing, figuratively, the horizontal bar of ``the
$\mathrm\Pi$''.  A transition rule that has such a constraint is called
\emph{bounded}.

\begin{definition}[bounded]
\label{def:ms-bounded}
Let $\GG$ be a MCS, $f$ a level mapping,\footnote{We
  sometimes write $f$ (for short) instead of $f_{\mu_L,\mu_H}$.} and
$g \in \GG$.  A transition rule $g=p(\bar x) \scgarrow \pi; q(\bar y)$
in $\GG$ is called \emph{bounded w.r.t.\ $f$} if $\pi\models
p_f^{high} \supsetpluseq p_f^{low}$.
\end{definition}

\vspace{-1mm}
\begin{definition}[orienting transition rules]
\label{orienting}
  Let $f$ be a level mapping.
  (1) $f$ \emph{orients} transition rule $g=p(\bar x) \scgarrow \pi;
  q(\bar y)\,$ if $\pi\models p_f^{high}(\bar x) \succsim
  q_f^{high}(\bar y)$ and $\pi\models p_f^{low}(\bar x) \precsim
  q_f^{low}(\bar y)$;
  (2) $f$ orients $g$ \emph{strictly} if, in addition, $\pi\models
  p_f^{high}(\bar x) \succ q_f^{high}(\bar y)$ or $\pi\models
  p_f^{low}(\bar x) \prec q_f^{low}(\bar y)$.
\end{definition}

\vspace{-1mm}
\begin{example} \label{ex-SCNP}
We refer to Ex.\ \ref{ex:mainexample} and the level mapping $f^1_{min,max}$
from Ex.\ \ref{ex:multiset-maps}.
Function $f^1_{min,max}$ orients all transition rules, where $g_1$
and $g_3$ are oriented strictly; $g_1$ and $g_4$ are 
bounded w.r.t.\ $f^1_{min,max}$
(the reader may be able to verify this by
observing the constraints, however later we explain how our algorithm obtains
this information).
\end{example}

\vspace{-1mm}
\begin{corollary}[of Def.~\ref{orienting} and
  Lemma~\arabic{lem:decdiff}] \label{cor:orientation} Let $f$ be a level
  mapping and define $\Phi_f(p(\bar x)) = p_f^{high}(\bar x) -
  p_f^{low}(\bar x)$.  If $f$ orients $g=p(\bar x) \scgarrow \pi;
  q(\bar y)\,$, then $\pi\models \Phi_f(p(\bar x))\succsim
  \Phi_f(q(\bar y))$; and if $f$ orients $g$ strictly, then
  $\pi\models \Phi_f(p(\bar x))\succ \Phi_f(q(\bar y))$.
\end{corollary}

The next theorem combines orientation and bounding to show how a level mapping
induces anchors.  Note that
we refer to cycles in the CFG also as ``cycles in $\GG$'', as the CFG
is implicit in $\GG$.

\setcounter{ranking thm}{\value{theorem}}
\addtocounter{ranking thm}{1}

\begin{theorem}
  Let $\GG$ be a MCS and $f$ a level mapping.  Let $g=p(\bar x)
  \scgarrow \pi; q(\bar y)\,$ be such that every cycle $\cal C$ including
  $g$
  satisfies these conditions: (1) all transitions in $\cal C$ are oriented
  by $f$, and at least one of them strictly; (2) at least one
  transition in $\cal C$ is bounded w.r.t.~$f$.
  Then $g$ is a $\Phi_f$-anchor for $\GG$, where $\Phi_f(p(\bar x)) =
  p_f^{high}(\bar x) - p_f^{low}(\bar x)$.
\end{theorem}

\begin{definition}[MCNP anchors and ranking functions] 
  \label{def:ranking lm} Let $\GG$ be a MCS and $f$ a level
  mapping. We say that $g$ is a MCNP-anchor for $\GG$ w.r.t.~$f$ if
  $f$ and $g$ satisfy the conditions of Thm.~\arabic{ranking thm}. 
%\pagebreak  
The
  function $\Phi_f$ is called a \emph{MCNP (intermittent) ranking
    function} (MCNP IRF).
\end{definition}

Note that if $g$ is not included in any cycle, then the definition is
trivially satisfied for any $f$.  Indeed, such transition rules are
removed by our algorithm without searching for level mappings at all.

\begin{example} 
The facts in Ex.~\ref{ex-SCNP} imply that $g_1$, $g_3$, and $g_4$ are
MCNP-anchors w.r.t.~$f^1_{min,max}$.
\end{example}

We remark that numerous termination proving techniques follow the
pattern of, repeatedly, identifying and removing anchors.  However,
typically, the function $\Phi$ used for ranking is required to be
strictly decreasing, and bounded, on the anchor itself, which (at
least implicitly) means that a lexicographic ranking function is being
constructed; see, e.g.,~\cite{CS:02}.  The anchor criterion expressed
in Thm.~\arabic{ranking thm} (inspired by~\cite[Thm.\ 8]{CADE07}) is more
powerful.  We note that the difference is only important with
non-well-founded domains. When the ranking is only done with orders
that are a priori well-founded, as for example in~\cite{JAR06,HM},
considering the strictly-oriented transitions as anchors is
sufficient.  In comparison to~\cite{CADE07}, we note that they do not
use the concept of anchors, and propose an algorithm which can
generate an exponential number of level-mapping-finding subproblems
(whereas ours generates, in the worst case, as many problems as there
are transition rules).

\section{The MCNP Problem}
\label{sec-sat}

In this section, we present necessary and sufficient conditions for
orientability and boundedness. Based on these, we conclude that
proving termination with MCNP IRFs is in NP. This also forms the
basis for our SAT-based algorithm in Sect.\ \ref{subsec:algorithm}.
 
\begin{definition}[MCNP]
  A system of monotonicity constraints is in MCNP if it has a
  ranking structure which is a tuple of MCNP IRFs.
\end{definition}

\noindent
It follows from Thm.\ \arabic{ranking structure implies termination},
that if a MCS is in MCNP, then it terminates.

\begin{example}\label{MCNP example}
  Consider again Ex.\ \ref{ex:mainexample} and the level mappings from
  Ex.\ \ref{ex:multiset-maps}.  Then, $\tuple{\Phi_{f^1},\Phi_{f^2}}$
  is a ranking structure for $\GG$. As already observed, $g_1, g_3$,
  and $g_4$ are MCNP-anchors for $f^1$.  Observe now that $f^2$ is
  both strict and bounded on $g_2$.
\end{example}

Ranking structures are constructed through iterative search for
suitable level mappings which prescribe pairs of (possibly
tagged) multisets of arguments which must satisfy relations of the
form $\succsim^\mu$, $\succ^\mu$, and $\supsetpluseq$.

Let $g=p(\bar x) \scgarrow \pi; q(\bar y)$ and $S$, $T$ be non-empty
sets of (tagged) argument positions of $p$ or of $q$.
We show how to check
for each $\mu \in \set{max,min,ms,dms}$ if $\pi\models S
\succsim^{\mu} T$. Viewing $g$ as a graph (as in Ex.\
\ref{ex:mainexample}), let $g^t$ denote the transpose of $g$ (obtained
by inverting the arcs).
While tagged level mappings can be represented as ``ordinary''
bi-multiset level mappings (as indicated in Sect.\ \ref{Multiset-based
  ranking functions}), for their SAT encoding, it is advantageous to
represent the orders on tagged pairs 
%\pagebreak
explicitly: 
\begin{equation}\label{eq:tagged1}
  \begin{array}{lcl}
    \pi\models (x,i)>(y,j) &\iff& 
       (\pi\models x> y)\lor ((\pi\models x\ge y) \land i>j) \\ 
\pi\models (x,i) \ge (y,j) &\iff& 
       (\pi\models x > y)\lor ((\pi\models x \ge y) \land i\ge j) 
  \end{array}
\end{equation}

Below, $x,y$ either both represent arguments, or both represent tagged
arguments, with relations $x>y$, $x\ge y$ interpreted accordingly.

\begin{enumerate}
\item \emph{max order: ($S \succsim^{max} T$)} every $y\in T$ must be
  ``covered'' by an $x\in S$ such that $\pi\models x \ge y$.  Strict
  descent requires $S\neq\emptyset$ and $x > y$.  
\item \emph{min order: ($S \succsim^{min} T$)} 
same conditions but on $g^t$ (now $T$ covers $S$).
\item \emph{multiset order: ($S \succsim^{ms} T$)} every $y\in T$ must
  be ``covered'' by an $x\in S$ such that $\pi\models x \ge y$.
  Furthermore each $x\in S$ either covers each related $y$ strictly
  ($x>y$) or covers at most a single $y$.  Descent is strict if there
  is some $x$ that participates in strict relations.
\item \emph{dual multiset order: ($S \succsim^{dms} T$)} 
same conditions but on $g^t$ (now $T$ covers $S$).
\end{enumerate}

\noindent
We also show how to decide if the relation $H\supsetpluseq L$ holds:
For $\mu_L,\mu_H\in\{max,min\}$ and $\mu_L=\mu_H$,  $H\supsetpluseq L$ holds iff
$\mu_H(H) \ge \mu_L(L)$.%
\footnote{Note that checking this amounts to checking for
  $\succsim^\mu$ in the case $\mu_L=\mu_H=\mu$; for the other cases,
  $max(H) \ge min(L)$ holds if there is at least one arc from an $H$
  vertex to an $L$ vertex; $min(H) \ge max(L)$ holds if there is an
  arc from every $H$ vertex to every $L$ vertex.}
For $\mu_L=min$ and $\mu_H\in\{ms,dms\}$,  $H\supsetpluseq L$  holds iff
$H\succsim^{min} L$.
      For $\mu_L\in\{ms,dms\}$ and $\mu_H=max$,  $H\supsetpluseq L$  holds iff
      $H\succsim^{max} L$.
      For $\mu_L=max$ and $\mu_H\in\{ms,dms\}$,  $H\supsetpluseq L$  holds if
      $min(H) \ge max(L)$.
      For $\mu_L\in\{ms,dms\}$ and $\mu_H=min$,  $H\supsetpluseq L$  holds if
      $min(H) \ge max(L)$.

      Since the above conditions allow for verification of a proposed
      MCNP ranking structure in polynomial time, we obtain the
      following theorem.

\setcounter{NP thm}{\value{theorem}}
\addtocounter{NP thm}{1}

\begin{theorem}
MCNP is in NP.
\end{theorem}

\section{A SAT-based MCNP Algorithm}
\label{subsec:algorithm}

Given that MCNP is in NP, we provide a reduction (an encoding) to SAT
which enables us to find termination proofs using an
off-the-shelf SAT solver. We invoke a SAT solver iteratively to generate
level-mappings and construct a  ranking structure
$\tuple{\Phi_1, \Phi_2, \dots, \Phi_m }$.
Our main algorithm is presented in Sect.~\ref{subsec:main-algorithm}.
Sect.~\ref{subsec:scc-algorithm} discusses how to find appropriate
level mappings and Sect.~\ref{A SAT encoding} introduces the SAT
encoding.

\subsection{Main algorithm}
\label{subsec:main-algorithm}

Given a MCS $\GG$, the idea is to iterate as follows: while $\GG$ is
not empty, find a
level mapping $f$ inducing
one or more anchors for $\GG$. Remove the anchors, and repeat.  The
instruction ``find a level mapping'' is performed using a SAT encoding
(for each of the compatible pairs of multiset orders).
To improve performance, the algorithm follows the SCC (strongly
connected components) decomposition of (the CFG of) $\GG$. This leads
to smaller subproblems for the SAT solver and is justified by the
observation that inter-component transitions are trivially anchors
(not included in any cycle). 
%\pagebreak  
In the following let $scc(\GG)$ denote
the set of non-vacant SCCs of $\GG$ (that is, SCCs which are not a vertex
without any arcs).

\newenvironment{SProg2}
     {\begin{small}\begin{tt}\begin{tabular}{ll}}%
     {\end{tabular}\end{tt}\end{small}}

\vspace{-2mm}
\paragraph{Main Algorithm.} ~\\
    \texttt{input:} $\GG$ (a MCS) \\
    \texttt{output:} $\rho=\tuple{f^1,f^2,\dots}$ (tuple of level mappings
             such that $\tuple{\Phi_{f^1},\Phi_{f^2},\dots}$ \\
             \hspace*{1.3cm} is a ranking structure for $\GG$).
             The algorithm aborts if $\GG$ is not in MCNP.
 \begin{enumerate}
 \item $\rho=\tuple{~}$ (empty queue); \quad  
       ${\cal S}=scc(\GG)$ (stack with non-vacant SCCs of $\GG$); 
 \item while (${\cal S}\neq\emptyset$) \vspace{-3mm}
   \begin{itemize}
   \item pop $\cal C$ from ${\cal S}$ (a MCS) and find (using SAT) a level
     mapping \\$f$ to anchor some transition rules in $\cal C$ \qin (if none,
     abort: ${\cal C}\notin$ MCNP)
   \item extend $f$ to program points $p$ not in $\cal C$ by 
     $f(p(\bar x))=\tuple{\emptyset,\emptyset}$
   \item append $f$ to $\rho$ and remove from $\cal C$ the $\Phi_f$-anchors that were found 
   \item push elements of $scc({\cal C})$ to ${\cal S}$
   \end{itemize}\vspace{-3mm}
  \item return $\rho$
 \end{enumerate}
       
\setcounter{thm:main}{\value{theorem}}
\addtocounter{thm:main}{1}

\vspace{-1mm}
\begin{theorem}
The main algorithm succeeds if and only if $\GG$ is in MCNP.
\end{theorem}

\subsection{Finding a level mapping}
\label{subsec:scc-algorithm}

The main step in the algorithm is to find a level mapping which
anchors some transition rules of a strongly-connected MCS.
Let $\GG$ be strongly connected and $f$ a level mapping which orients
all transition rules in $\GG$, strictly orients the transition rules
from a non-empty set $S \subseteq \GG$, and where $B \subseteq \GG$
(non-empty) are bounded. Following Thm.~\arabic{ranking thm}, a
transition rule $g$ is an anchor if every cycle in $\GG$ containing
$g$ has an element from $S$ and an element from $B$. We need to check
all cycles in $\GG$ (possibly exponentially many). We describe a way
of doing so by numbering nodes which lends itself well to a SAT-based
solution.

\newcommand{\nn}{\mathit{num}}
\begin{definition}[node numbering]\label{def:node numbering}
  A \emph{node numbering} is a function $\nn$ from $n$ program points
  to $\set{1,\dots,n}$.  For $g=p(\bar x) \scgarrow \pi; q(\bar y)$,
  we denote $\Delta\nn(g) = num(q)-num(p)$.  For a set ${\cal
    H}\subseteq \GG$, we say that $\nn$ \emph{agrees with $\cal H$} if
   for all $g \in \GG$:\ \
  $\Delta\nn(g) > 0 \Rightarrow g\in {\cal H}$.
\end{definition}

Now for $g\in\GG$, checking that every cycle of $\GG$ containing $g$
also contains an element of $S$, is reduced to finding a node
numbering $\nn_S$ with ${\Delta\nn_S(g)\ne 0}$ which agrees with
$S$. Then, any cycle containing $g$ must contain also an edge $g'$ with
${\Delta\nn_S(g')>0}$. But this implies that $g'\in S$ because $\nn_S$
agrees with $S$.

\setcounter{anchor correct}{\value{theorem}}
\addtocounter{anchor correct}{1}

\begin{lemma}
  Let $\GG$, $f$, $S$, and $B$ be as above.  Then, $g\in\GG$ is a MCNP-anchor 
  for $\GG$ w.r.t~$f$ if and only if: (1) $g\in S\cap B$; or 
  (2) there are node numberings $\nn_S$ and $\nn_B$ agreeing with $S$
  and $B$ respectively, such that ${\Delta\nn_S(g)\ne 0}$ and
  $\Delta\nn_B(g)\ne 0$.
\end{lemma}

\vspace{-2mm}
\begin{example} \label{ex:mainalgorithm} We now describe the
  application of the Main Algorithm to Ex.~\ref{ex:mainexample}.
  Initially, there is a single SCC, ${\mathcal C} = \GG$.  Using SAT
  solving (as described in Sect.~\ref{A SAT encoding}) we find that
  level 
%\pagebreak 
mapping $f^1$ of Ex.~\ref{ex:multiset-maps} orients all
  transitions, strictly orients $S = \{g_1,g_3\}$ and is bounded on $B
  = \{g_1,g_4\}$. Hence, by choosing the numbering $\nn_B(p)=2$,
  $\nn_B(q)=1$, $\nn_S(p)=1$, $\nn_S(q)=2$, we obtain that $g_1$,
  $g_3$ and $g_4$ are anchors.  Note that the problem encoded to
  SAT represents the choice of the level mapping and node numbering at
  once.
  Now, $\rho$ is set to $\tuple{f^1}$, and the anchors are removed
  from $\mathcal C$, leaving a SCC consisting of point $p$ and
  transition rule $g_2$. In a second iteration, level mapping $f^2$ of
  Ex.~\ref{ex:multiset-maps} is found and appended to $\rho$.  No SCC
  remains, and the algorithm terminates.

 Note that our algorithm is non-deterministic (due to leaving some
 decisions to the SAT solver).  In this example, the first iteration
 could come up with the numbering $\nn_B(p)=\nn_B(q)=1$, which would
 cause only $g_1$ to be recognized as an anchor. Thus, another
 iteration would be necessary, which would find a numbering according
 to which $g_3$ and $g_4$ are anchors, since this time there is no
 other option.
\end{example}

\subsection{A SAT encoding}\label{A SAT encoding}

Let $\GG$ be a strongly connected MCS (assume the context of the
Main Algorithm of Sect.~\ref{subsec:main-algorithm}).  For a compatible
pair $\mu_L,\mu_H$ we construct a propositional formula
$\Phi^\GG_{\mu_L,\mu_H}$ which is satisfiable iff there exists a level
mapping $f_{\mu_L,\mu_H}$ that anchors some transition rules in $\GG$. 
We focus on tagged level mappings (omitting tags is the same as
assigning them all the same value).

Each program point $p$ and argument position $i$ is associated with an
integer variable $tag_p^i$.
Integer variables are encoded through their bit representation.
In the following, we write, for example, $||n>m||$ to
indicate that the relation $n>m$ on integer variables is encoded to a
propositional formula in CNF.  
Let $g=p(\bar x) \scgarrow \pi ; q(\bar y)$ and consider each $a,b \in
\bar x \cup \bar y$. At the core of the encoding, we use a formula
$\varphi_{rel}^g$ which introduces a propositional variable $e^g_{a>b}$ to
specify a corresponding ``tagged edge'', $e^g_{a>b} \leftrightarrow
\pi \models (a,tag_1) > (b,tag_2)$, as prescribed in Eq.~(\ref{eq:tagged1}).
 Here, $tag_1$ and $tag_2$ are the
integer tags associated with the program points and argument positions
of $a$ and $b$ (in $g$). We proceed likewise for the propositional variable
$e^g_{a\ge b}$. 

   \begin{example}
     Consider $g_3=p(x_1,x_2,x_3) \scgarrow y_1>x_1, x_2\geq y_2;
     q(y_1,y_2)$ from Ex.~\ref{ex:mainexample}.  The formula
     $\varphi_{rel}^{g_3}$ contains (among others) the following
     conjuncts.
     From $(y_1>x_1)$, $(e^{g_3}_{y_1>x_1}\leftrightarrow
     \mathtt{true})$ and $(e^{g_3}_{y_1\ge x_1}\leftrightarrow
     \mathtt{true})$; from $(x_2\geq y_2)$, $(e^{g_3}_{x_2>y_2}
     \leftrightarrow ||tag^2_p>tag^2_q||)$ and $(e^{g_3}_{x_2\geq
       y_2} \leftrightarrow ||tag^2_p\geq tag^2_q||)$.
     Observe also,
     $e^{g_3}_{x_1>y_2}\leftrightarrow \mathtt{false}$ and $e^{g_3}_{x_1\geq
       y_2}\leftrightarrow \mathtt{false}$.
   \end{example}

\noindent
We introduce the following additional propositional variables:

\hspace{-5mm}
\begin{minipage}[t]{0.56\linewidth}
  \begin{itemize}

  \item $weak^g \Leftrightarrow g$ oriented weakly by
    $f_{\mu_L,\mu_H}$
  \item $strict^g \Leftrightarrow g$ oriented strictly by
    $f_{\mu_L,\mu_H}$
  \item $bound^g \Leftrightarrow p_f^{high}(\bar x) \supsetpluseq
    p_f^{low}(\bar x)$
  \item $anchor^g \Leftrightarrow g$ is an anchor w.r.t.~$f$
    in $\GG$
  \end{itemize}
\end{minipage}
\hspace{-7mm}
\begin{minipage}[t]{0.49\linewidth}
\begin{itemize}
\item $weak^g_{low}\Leftrightarrow q_f^{low}(\bar y) \succsim^{\mu_L}
  p_f^{low}(\bar x)$
\item $strict^g_{low}\Leftrightarrow q_f^{low}(\bar y) \succ^{\mu_L}
  p_f^{low}(\bar x)$
\item $weak^g_{high}\hspace{-1mm}\Leftrightarrow\hspace{-1mm}
  p_f^{high}(\bar x) \succsim^{\mu_H} q_f^{high}(\bar y)$
\item $strict^g_{high}\hspace{-1mm}\Leftrightarrow\hspace{-1mm}
  p_f^{high}(\bar x) \succ^{\mu_H} q_f^{high}(\bar y)$
\end{itemize}
\end{minipage}	

\vspace*{.2cm}
\noindent
and, for every program point $r$, two integer variables
$\nn_S^r$ and $\nn_B^r$ to represent the node numberings from Def.\
\ref{def:node numbering}.  

Our encoding takes the following form:
\[\small
  \Phi_{\mu_L,\mu_H}^\GG = \left(\bigwedge_{g\in\GG} weak^g\right) \land
                 \left(\bigvee_{g\in\GG} anchor^g \right)\land
    \left(\begin{array}[c]{l}
    \varphi_{rel}^\GG\wedge\psi^\GG\wedge\psi^\GG_{pos}\wedge\psi^\GG_{low}
    \wedge \\
     \wedge ~\psi^\GG_{high} \wedge \psi^\GG_{bound}\wedge \psi^\GG_{ne}
    \end{array}\right)
\]
The first two conjuncts specify that $f_{\mu_L,\mu_H}$ is a level
mapping which orients $\GG$, the third is specified as
$\varphi_{rel}^\GG = \bigwedge_{g\in\GG}\varphi_{rel}^g$, and the rest
are explained below:

\vspace{-2mm}
\paragraph{Proposition  $\psi^\GG$} imposes the intended meanings on
$weak^g$, $strict^g$ and $anchor^g$ (see Def.\ \ref{orienting} and Lemma \arabic{anchor
  correct}).

\vspace{-6mm}
\[\small
  \psi^{\GG} = \hspace{-6mm}
  \bigwedge_{g=\; p(\bar x)\scgarrow \pi ; \; q(\bar y)}
  \left(
  \begin{array}{l}
      \displaystyle
       weak^g \leftrightarrow (weak^g_{low} \wedge weak^g_{high})\quad
  \wedge\\
      \displaystyle
      strict^g \leftrightarrow 
           (weak^g \wedge (strict^g_{low} \vee strict^g_{high}))\quad
  \wedge\\
      \displaystyle
  anchor^g \leftrightarrow\hspace{-2mm}
  \begin{array}[t]{l}
    ((p\neq q) \wedge (||num_S^p\neq num_S^q|| \wedge ||num_B^p\neq num_B^q||))
    ~\vee\\
    ((p= q) \wedge strict^g \wedge bound^g)
  \end{array}

   \end{array}\hspace{-3mm}
 \right)\]

\vspace{-3mm}
\paragraph{Proposition $\psi^\GG_{pos}$} enforces that the node numberings
$num_S$ and $num_B$
agree with sets $S$ and $B$, cf.\ Lemma \arabic{anchor correct}:
\vspace{-2mm}
\[\small
  \psi^\GG_{pos} = \hspace{-4mm}
  \bigwedge_{g=\;p(\bar x)\scgarrow \pi ; \; q(\bar y)}\hspace{-2mm}
  \left(
  \begin{array}{l}\displaystyle
  (||num_S^p<num_S^q|| \rightarrow strict^g) ~\wedge \\(||num_B^p<num_B^q|| \rightarrow bound^g)
   \end{array}
 \right)
\]

\vspace{-4mm}
\paragraph{Proposition $\psi^\GG_{high}$} imposes that $weak^g_{high}$
and $strict^g_{high}$ are {\tt true} exactly when $p_f^{high}(\bar x)
\succsim^{\mu_H} q_f^{high}(\bar y)$ and $p_f^{high}(\bar x)
\succ^{\mu_H} q_f^{high}(\bar y)$, respectively.
We focus on the case when $\mu_H=max$, the other cases are similar and
omitted for lack of space.  The encoding of proposition
$\psi^\GG_{low}$ is similar (and also omitted for lack of space).
\[\small
  \psi_{high}^\GG = \hspace{-6mm}
  \bigwedge_{g=\; p(\bar x)\scgarrow \pi ; \; q(\bar y)}
  \left(
  \hspace{-2mm}\begin{array}{l}\displaystyle
  weak^g_{high} \leftrightarrow \hspace{-2mm}
     \bigwedge_{1\leq j\leq m} 
\!\!     \left(
        q_j^{high} \rightarrow \hspace{-2mm}\bigvee_{1\leq i\leq n} 
                      \hspace{-2mm}(p_i^{high}\wedge e^g_{x_i\geq y_j})
     \right) \wedge\\
  \displaystyle
  strict^g_{high} \leftrightarrow \hspace{-2mm}
     \bigwedge_{1\leq j\leq m} 
\!\!     \left(
        q_j^{high} \rightarrow \hspace{-2mm}\bigvee_{1\leq i\leq n} 
                      \hspace{-2mm}(p_i^{high}\wedge e^g_{x_i>y_j})
     \right)
     \land\bigvee_{1\leq i\leq n} \hspace{-2mm}p_i^{high} 
   \end{array}
\!\!\! \right)
\]
The propositional variables $p_i^{low}$, $p_i^{high}$, $q_j^{low}$, and
$q_j^{high}$ ($1\le i\le n, 1\le j\le m$) indicate the argument
positions of $p/n$ and $q/m$ selected by the level mapping
$f_{\mu_L,\mu_H}$ for the low and high sets, respectively.
The first subformula specifies that a transition rule is weakly
oriented by the $max$ order if for each $j$ where $q_j^{high}$ is
selected (i.e., the $j$-th argument of $q$ is in $q^{high}$), at least
one of the selected positions $p_i^{high}$ has to ``cover''
$q_j^{high}$ with a weak constraint $x_i\geq y_j$. The second
subformula is similar for the case of strict orientation with the
additional requirement that at least one $p_i^{high}$ should be
selected.

\vspace{-3mm}
\paragraph{Proposition $\psi^\GG_{bound}$} constrains $bound^g$ to be
{\tt true} iff $p_f^{high} \supsetpluseq p_f^{low}$ is satisfied by
$g$.  As observed in Sect.\ \ref{sec-sat}, this test boils down to
four cases. We illustrate the encoding for the case
$min(p_f^{high}(\bar x)) \geq max(p_f^{low}(\bar x))$:
\[ \small
   \psi^\GG_{bound} = \hspace{-4mm}
      \bigwedge_{g=\;p(\bar x)\scgarrow \pi ; \; q(\bar y)}
      \left(
      bound^g \leftrightarrow
         \bigwedge_{1\leq i\leq n,1\leq j\leq n} 
         \left(
            (p_i^{high} \wedge p_j^{low}) \rightarrow e^g_{x_i\geq x_j}
         \right) 
     \right)
    \]

\paragraph{Proposition $\psi^\GG_{ne}$} 
constrains the level mapping so that for each program point $p$, the
sets $p^{low}$ and $p^{high}$ are not empty.  Let $\PP$ denote the set
of program points in $\GG$.%

\vspace{-2mm}
{\begin{small}
 \[ \psi^\GG_{ne} = \hspace{-1mm}
      \bigwedge_{p\in\PP}
      \left(
         \left(
         \bigvee_{1\leq i\leq n} 
            {p_i^{low}}
         \right) 
	\wedge
         \left(
         \bigvee_{1\leq i\leq n} 
            {p_i^{high}}
         \right) 
     \right)
    \]\end{small}}

\vspace{-4mm}
\section{Implementation and Experiments}\label{sec:exp}

We implemented a 
termination analyzer based on our SAT
encoding for MCNP and tested it on three benchmark suites. Experiments
were conducted running the \textsf{SAT4J} 
\cite{sat4j} solver on an Intel Core i3 at 2.93 GHz with 2 GB RAM.
For further details on our experiments see
\ref{app:experiments} and 
\url{http://aprove.informatik.rwth-aachen.de/eval/MCNP}.

\vspace{-3mm}
\paragraph{Suite 1}
 consists of 81 MCSs obtained from various
research papers on termination and from abstracting
textbook style C programs.\footnote{Using a translator developed by
  A.\ Ben-Shabtai and Z.\ Mann at Tel-Aviv Academic College.}
MCNP proves 66 of them terminating with an average runtime of 0.55s
(maximal runtime is 5.15s). This suite contains the 32 examples from
the evaluation of \cite{RTA09}. That paper introduced \emph{integer
  term rewrite systems} (ITRSs), where standard operations on integers
are pre-defined, and showed how to use a rewriting-based termination
prover like \aprove\ for algorithms on integers. MCNP shows
termination of 27 of these.  \aprove\footnote{Using an Intel Core 2
  Quad CPU Q9450 at 2.66 GHz with 8 GB RAM.}  proves termination of
these 27 and one more example.  On the 32 examples from \cite{RTA09},
the average runtime of MCNP is 0.22s, whereas the average runtime of
\aprove\ is 5.3s for the examples with no timeout (\aprove\ times out
after 60s on 4 examples).
This shows that MCNP is sufficiently powerful for representative
programs on integers and demonstrates the efficiency of our SAT-based
implementation.  The comparison with \aprove\ on the examples from
\cite{RTA09} indicates that MCNP has about the same precision and is
significantly faster.

\vspace{-3mm}
\paragraph{Suite 2}
originates from the Java Bytecode (JBC) programs in the \emph{JBC} and
\emph{JBC Recursive} categories of the \emph{International Termination
  Competition} 2010.\footnote{In this competition, \aprove, \costa,
  and \textsf{Julia} competed against each other.  \\ See
  \url{http://www.termination-portal.org/wiki/Termination_Competition}
  for details.}
165 MCS instances were obtained by first applying the preprocessor of
the termination analyzer \costa\ \cite{FMOODS08} resulting in
(binary clause) constraint
logic programs with linear constraints (CLPQ).
After minor processing, these are abstracted to MCSs (applying SWI
Prolog with its CLPQ library).
MCNP provides a termination proof for 92 of these 
with an average runtime of 0.66s (maximal runtime is 16.31s).
In contrast, \costa\footnote{Experiments for \costa\ were performed on
  an Intel Core i5 at 3.2 GHz with 3 GB RAM.}\ shows termination of
102 programs. However, it encounters a (120 second) timeout on 5
instances. \costa's average runtime for the examples with no timeout
is 0.076s.
From these experiments we see that although MCNP is based on very
simple ranking functions, it is able to provide many of the proofs,
and does not encounter timeouts.  Moreover, there are 5 programs where
MCNP provides a proof and \costa\ does not (4 due to timeouts).

\vspace{-3mm}
\paragraph{Suite 3.} Here,
the Competition 2010 version of the termination analyzer \aprove\
abstracts JBC programs from the (non-recursive) \emph{JBC} category of
the Termination Competition 2010 to ITRSs. (This abstraction from
\cite{JBC-Correctness,RTA10} only works for programs without
recursion.)
To further transform ITRSs into MCSs, we apply an abstraction which
maps terms to their size and replaces non-linear arithmetic
sub-expressions by fresh variables. This results in  a  CLPQ representation
which is further abstracted to MCSs as for Suite 2.
For the resulting 127 instances, MCNP provides 63 termination proofs,
8 timeouts after 60s, and an average runtime of 5.76s
(we count timeouts as 60s).
%%%%%
To compare, we apply \aprove\ directly\footnote{Using an Intel Xeon
  5140 at 2.33 GHz with 16 GB RAM and imposing a time limit of 60s.}
but fix the abstraction to be the same as in the preprocessor for
MCNP. This results in 73 termination proofs and 8 timeouts with an
average time of 14.16s. There are 5 instances where MCNP
provides a proof not found by \aprove.
Applying \aprove\ without fixing the abstraction gives 95 termination
proofs, 19 timeouts, and an average time of 17.12s (there are still 3
instances where MCNP provides a proof not found by \aprove).
%%%%%% 
%%%%%%
This shows that the additional proving power in \aprove\ comes
primarily from the search for the right abstraction. Once fixing the
abstraction, MCNP is of similar precision and much
faster.
%%%%%% 
Thus, it could be fruitful to use a combination of tools where the
MCNP-analyzer is tried first and the rewrite-based analyzer is only
applied for the remaining ``hard'' examples.

 \section{Conclusion}
 We introduced a new approach to prove termination of
 monotonicity-constraint transition systems. The idea is to construct
 a ranking structure, of a novel kind, extending previous work in this
 area.  To verify whether a MCS has such a ranking structure, we use
 an algorithm based on SAT solving. We implemented our algorithm and
 evaluated it in extensive experiments. The results demonstrate the
 power of our approach and show that its integration into termination
 analyzers for Java Bytecode advances the state of the art of
 automated termination analysis.

\medskip

\noindent
\textbf{Acknowledgment.} We thank Samir Genaim for help with the  benchmarking.

 \appendix
\newpage
\section{Proofs}
\label{app:proofs}

\setcounter{theorem}{\value{ranking structure implies termination}}
\addtocounter{theorem}{-1}
\begin{theorem} 
If there is a ranking structure for $\GG$, then $\GG$ terminates.
\end{theorem}
\begin{proof}
  Suppose that $\GG$ has an infinite run $\tilde s = p_0(\bar
  x_0)\stackrel{\pi_0}{\to} p_1(\bar x_1)\stackrel{\pi_1}{\to}
  p_2(\bar x_2) \dots$.  Let $\cal H$ be the set of transition rules
  that are applied infinitely often in this run.  Using the notation
  of Def.~\ref{def:structure}, choose $g\in \mathcal H$ such that
  $i(g)$ is minimal. Then $g$ is a $\Phi_{i(g)}$-anchor for a subset
  of $\GG$ containing ${\mathcal H}$.  Consider the infinite tail of
  $\tilde s$ that stays within $\cal H$ and note that it includes
  infinitely many occurrences of $g$. Using
  Def.~\ref{def:anchor}, it is not hard to show that there is an
  infinite sequence $i_1<i_2<i_3<\cdots$, such that for all $k>0$,
  $\Phi_{i(g)}(p_{i_k}(\bar x_{i_k})) \in {\mathcal D}_+$, and in
  addition, $\Phi_{i(g)}(p_{i_k}(\bar x_{i_k})) \succ
  \Phi_{i(g)}(p_{i_{k+1}}(\bar x_{i_{k+1}}))$.  This contradicts the
  well-foundedness of ${\mathcal D}_+$, thus we conclude that such an
  infinite run cannot exist.
\end{proof}

\setcounter{theorem}{\value{underlying domain well founded}}
\addtocounter{theorem}{-1}
\begin{lemma}
  For all $\mu\in \{max, min, ms, dms\}$,
  $(\wp_n(\ints),\succsim^\mu)$ is a total quasi-order, with
  $\succ^\mu$ its strict part; and
  \mbox{$(\wp_n(\ints),\succsim^\mu)_+$} is well-founded. 
\end{lemma}

\begin{proof}
  The claims are straightforward for the $max$ and $min$ orders.  For
  the multiset orders, since our value domain ($\ints$) is totally
  ordered, we will justify the claims by referring to properties of
  the lexicographic order. Let $S, T \in \wp_n(\ints)$.  For the
  multiset order ($ms$), let $tup(S)$ be the tuple consisting of the
  elements of $S$ in non-increasing order.  If $S\ne T$, then either
  one tuple is a prefix of another (then the larger multiset is also
  greater under $\succ^{ms}$), or there is a first position where the
  elements differ. If in this first position the element of $S$ is
  larger, it is easy to show that $S\succ^{ms} T$. Thus,
  $\succsim^{ms}$ agrees with the lexicographic ordering on the
  tuples, which proves that it is a total quasi-order (in fact, a
  total order).
 
  Multisets in $(\wp_n(\ints),\succsim^{ms})_+$ map to tuples of
  non-negative integers; it is well-known that the lexicographic order
  on tuples of non-negative integers is well-founded.
 
  For $\succ^{dms}$ we argue in the same way, using tuples in
  non-decreasing order.
\end{proof}

\setcounter{theorem}{\value{lem:decdiff}}
\addtocounter{theorem}{-1}
\begin{lemma} 
  Let $L, H$ be two multisets of compatible types $\mu_L, \mu_H$, and
  let $\mu_D$ be the type of $H-L$.  Let $L', H'$ be of the same types
  as $L, H$ respectively.  Then
\begin{eqnarray*}
H\succsim^{\mu_H} H' \land L \precsim^{\mu_L} L'   
             &\Longrightarrow& H-L \succsim^{\mu_D} H'-L'  ;\\
H\succ^{\mu_H} H' \land L \precsim^{\mu_L} L'
             & \Longrightarrow&  H-L \succ^{\mu_D} H'-L'  ;\\
H\succsim^{\mu_H} H' \land L \prec^{\mu_L} L'
             & \Longrightarrow&  H-L \succ^{\mu_D} H'-L'  \,.
\end{eqnarray*}
\end{lemma}

In order to prove
 Lemma~\arabic{lem:decdiff} we first need the following definition and lemma.

\begin{definition}[multiset negation]
\label{def:ms-neg}
Let $S=\set{s_1,s_2,\dots,s_n}$ be a multiset of integers.
The negation of $S$, $(-S)$, is $\set{-s_1,-s_2,\dots,-s_n}$.
\end{definition}

\begin{lemma}
\label{lemma:dual-ord}
Let $S,T$ be non-empty multisets.
\begin{enumerate}
\item If $S \succsim^{max} T$ then $(-T) \succsim^{min} (-S)$ and if $S \succ^{max} T$ then $(-T) \succ^{min} (-S)$.
\item If $S \succsim^{min} T$ then $(-T) \succsim^{max} (-S)$ and if $S \succ^{min} T$ then $(-T) \succ^{max} (-S)$.
\item If $S \succsim^{ms} T$ then $(-T) \succsim^{dms} (-S)$ and if $S \succ^{ms} T$ then $(-T) \succ^{dms} (-S)$.
\item If $S \succsim^{dms} T$ then $(-T) \succsim^{ms} (-S)$ and if $S \succ^{dms} T$ then $(-T) \succ^{ms} (-S)$.
\end{enumerate}
\end{lemma}

\begin{proof}
\label{proof:dual-ord}
We only prove (3), since (1) and (2) are trivial and (4) is similar to
(3). $S \succsim^{ms} T \wedge S \nsucc^{ms} T$ holds iff $S=T$ and in
this case $(-T) \succsim^{dms} (-S)$ by the definition. Let $S
\succ^{ms} T$. We need to prove that $(-T) \succ^{dms} (-S)$.

Let $C=S \cap T$, $S_{rest}=S \setminus C$ and $T_{rest}=T \setminus
C$. Now we can express $(-S)$ and $(-T)$ in the following way:
$(-S)=(-C) \cup (-S_{rest})$ and $(-T)=(-C) \cup (-T_{rest})$.  By the
definition of $\succ^{ms}$, $S_{rest} \succ^{max} T_{rest}$. So
$(-T_{rest}) \succ^{min} (-S_{rest})$. According to the definition of
$\succ^{dms}$ we conclude that $(-T) \succ^{dms} (-S)$.
\end{proof}

Next we prove Lemma~\arabic{lem:decdiff}.
\begin{proof*}
  The following properties are easy to prove:
	\begin{description}
        \item[(i)] If the elements of two multisets $S$ and $T$ can be
          put in one-to-one correspondence $(s_i,t_i)$ such that
          $s_i\ge t_i$ in all pairs, then $S\succsim^\mu T$ for all
          $\mu$. If for all pairs $s_i> t_i$, then $S\succ^\mu T$.
	\item[(ii)] If $H,H'$ are multisets and $c\in\ints$, then
          shifting all elements of both sets by $c$ preserves the
          order relations among them.
	\end{description}

Now we will prove the lemma for each of the cases.
{
\begin{enumerate} \leftmargin 0pt \labelwidth 1em \itemindent 1em  
\item {$\mu_L=max$:} According to property (ii) we have
  \[H\succsim^{\mu_H} H' \,\Rightarrow\, \sset{ h-max(L')}{h\in H }
  \succsim^{\mu_H} \sset{ h'-max(L')}{h'\in H' }\] That is,
  $H-L'\succsim^{\mu_H} H'-L'$. In the same way we can see that
  $H\succ^{\mu_H} H' \Rightarrow H-L'\succ^{\mu_H} H'-L'$.

  Since $max(L') \geq max(L)$, according to property (i) we have
  $H-L\succsim^{\mu_H} H-L'$ and if $max(L') > max(L)$ then
  $H-L\succ^{\mu_H} H-L'$.

  By transitivity,
  \[H-L\succsim^{\mu_H} H-L' \quad\wedge\quad H-L\succsim^{\mu_H}
  H'-L' \quad\Rightarrow\quad H-L\succsim^{\mu_H} H'-L'\] and if one
  of the orderings is strict then $H-L\succ^{\mu_H} H'-L'$.

\item {$\mu_L=min$:} The proof is similar to (1).

\item {$\mu_L=ms,\mu_H=min$:} Given $L \precsim^{ms} L'$ and $H
  \succsim^{min} H'$, according to Lemma~\ref{lemma:dual-ord} we have
  $(-L) \succsim^{dms} (-L')$ and $(-H) \precsim^{max}
  (-H')$. According to part (1) of the proof, we obtain $(-L-(-H))
  \succsim^{dms} (-L'-(-H'))$.

  Moreover, by Def.\ \ref{def:ms-diff} (1), $(-L-(-H))=\sset{
    (-\ell)-max(-H)}{(-\ell)\in (-L) }=\sset{ min(H)-\ell}{\ell\in L
  }=H-L$ by Def.\ \ref{def:ms-diff} (2).  Similarly
  $(-L'-(-H'))=H'-L'$. So $H-L \succsim^{dms} H'-L'$.

  We can easily see that if $L \prec^{ms} L'$ or $H \succ^{min} H'$
  then $H-L \succ^{dms} H'-L'$.

\item {$\mu_L=ms,\mu_H=max$:} The proof is similar to (3).
\item {$\mu_L=dms,\mu_H=min$:} The proof is similar to (3).
\item {$\mu_L=dms,\mu_H=min$:} The proof is similar to (3).
\end{enumerate}}
\end{proof*}

\setcounter{theorem}{\value{ranking thm}}
\addtocounter{theorem}{-1}
\begin{theorem}
  Let $\GG$ be a MCS and $f$ a level mapping.  Let $g=p(\bar x)
  \scgarrow \pi; q(\bar y)\,$ be such that every cycle $\cal C$ including
  $g$ satisfies these conditions: (1) all transitions in $\cal C$ are
  oriented by $f$, and at least one of them strictly; (2) at least one
  transition in $\cal C$ is bounded w.r.t.~$f$.
Then $g$ is a $\Phi_f$-anchor in $\GG$, where 
$\Phi_f(p(\bar x)) = p_f^{high}(\bar x) - p_f^{low}(\bar x)$.
\end{theorem}

\begin{proof}
Consider a run $p_0(\bar
    x_0)\stackrel{\pi_0}{\to} p_1(\bar x_1)\stackrel{\pi_1}{\to} \ldots
    \stackrel{\pi_{k-1}}{\to} p_k(\bar x_k) \stackrel{\pi_{k}}{\to}
    p_{k+1}(\bar x_{k+1}) $ where both $p_0(\bar x_0)\stackrel{\pi_0}{\to}
    p_1(\bar x_1)$ and $p_k(\bar x_k) \stackrel{\pi_{k}}{\to} p_{k+1}(\bar
    x_{k+1})$ correspond to the transition rule $g$.
By assumption (1) of the theorem, and
Corollary~\ref{cor:orientation}, 
$\Phi_f(p_{i}(\bar x_{i}))\succsim  \Phi_f(p_{i+1}(\bar x_{i+1}))$ for all $0\le i \le k$, and, moreover, at least one of these inequalities is strict.
By assumption (2), and Def.~\ref{def:ms-bounded},  we have
$\Phi_f(p_{i}(\bar x_{i}))\in {\mathcal D}_+$ for some $0\le i\le k$.

We conclude that $g$ is a $\Phi_f$-anchor for $\GG$.
\end{proof}

\setcounter{theorem}{\value{NP thm}}
\addtocounter{theorem}{-1}
\begin{theorem}
MCNP is in NP.
\end{theorem}

\noindent {\normalfont\itshape Proof} \par\noindent Let $\GG$ be an MC
system.  If it is in MCNP, there is a ranking structure of polynomial
size (see Def.~\ref{def:irredundant} and subsequent comment).  The
following evidence suffices for verifying the ranking structure:
\begin{enumerate}
\item The list of level mappings, given explicitly: that is, for each
  program point, the high and low sets are listed.
\item For each level mapping $f^i$, the transition rules claimed to be
  oriented or strictly oriented by $f^i$ and those that are claimed to
  be bounded with respect to it; and additional information used to
  verify that these conditions hold.
\end{enumerate}

The additional information mentioned last consists of the set of arcs,
from the MC graph representation, that proves the desired relation
among multisets, according to the observations given in Sect.\ \ref{sec-sat}.
For example, to prove $\pi\models S \succsim^{max} T$, we
require a list of pairs $(x,y)$ with $x\in S$ and $y\in T$ that
satisfy $\pi\models x\ge y$, and include all $y\in T$.

This information has polynomial size and can be verified in polynomial
time by the following algorithm.  First, locally, (strict) orientation
and boundedness are verified with the aid of the supplied information.
Secondly, a counter $i$ is initialized to 1.  The $\Phi_{f^i}$ anchors
are found, according to Thm.~\arabic{ranking thm}, by a polynomial-time
graph algorithm (based on depth-first search).
Then they are removed, $i$ is incremented, and the
procedure is repeated. When the list is exhausted, $\GG$ should be
vacant; otherwise, the verification fails.
{\hspace*{1em}\hbox{\proofbox}}

\setcounter{theorem}{\value{thm:main}}
\addtocounter{theorem}{-1}
\begin{theorem}
The main algorithm succeeds if and only if $\GG$ satisfies MCNP.
\end{theorem}

\begin{proof}
  If the algorithm succeeds, returning $\rho = \tuple{f^1,f^2,\dots}$,
  then $\tuple{\Phi_{f^1},\Phi_{f^2},\dots}$ is a ranking structure
  for $\GG$: this is immediate from the definition of a ranking
  structure, provided the correctness of the sub-procedures that
  identify anchors.

  In the other direction, we assume that
  $\tuple{\Phi_{f^1},\Phi_{f^2},\dots}$ is a ranking structure for
  $\GG$, and prove that the algorithm succeeds.
 
  Consider any iteration of the main loop, and let $\cal C$ be the SCC
  popped from the stack. We claim that there exists an MCNP IRF for
  $\cal C$: indeed, using the notation of Def.~\ref{def:structure},
  choose $g\in \mathcal C$ such that $i(g)$ is minimal. Then
  $\Phi_{f^{i(g)}}$ anchors $g$ for a subset of $\GG$ that contains
  $\cal C$. Our search procedure will find an MCNP IRF (though not
  necessarily the same), and will remove one or more anchors. Thus, at
  the completion of each iteration, a non-empty set of transition
  rules has been removed from $\mathcal C$. The contents of the stack
  are, therefore, a set of SCCs which are strictly reduced (with
  respect to the number of arcs) in each iteration, which proves that
  the algorithm terminates. It will not abort, as we have just argued
  that the search for a level mapping and anchors must succeed.
\end{proof}

\setcounter{theorem}{\value{anchor correct}}
\addtocounter{theorem}{-1}
\begin{lemma}
   Let $\GG$, $f$, $S$, and $B$ be as in Sect.\ \ref{subsec:scc-algorithm}.  Then,
   $g\in\GG$ is a MCNP-anchor  
  for $\GG$ w.r.t~$f$ if and only if: (1) $g\in S\cap B$; or 
  (2) there are node numberings $\nn_S$ and $\nn_B$ agreeing with $S$
  and $B$ respectively, such that ${\Delta\nn_S(g)\ne 0}$ and
  $\Delta\nn_B(g)\ne 0$.
\end{lemma}

\begin{proof}
  Let $g = p(\bar x) \scgarrow \pi; q(\bar y)$.
If $p=q$, it is easy to see that $g$ is an anchor w.r.t.~$f$ if and
  only if $g\in S\cap B$. Case (2) is impossible if $p=q$. Next, let
  $p\ne q$.

  First, suppose that a node numbering as required does exist.  Now if
  $\mathcal{C}$ is a cycle including $g$, the $\nn_S$ values on this cycle are
  not all equal; so there must be a $g'=p'(\bar x') \scgarrow \pi';
  q'(\bar y') \in \mathcal{C}$ for which $\nn_S(p') > \nn_S(q')$.  Every
  transition rule with such numbering was required to be in $S$.  A
  similar argument shows that $\cal C$ must include a bounded transition
  rule.  Thus, $g$ satisfies the requirements in Thm.~\arabic{ranking
    thm}, justifying the ``if'' part of the lemma.
  
  For ``only if,'' suppose that $g$ is an anchor.  Let ${\GG}_B =
  {\GG}\setminus B$.  Assign numbers to the strongly-connected
  components of ${\GG}_B$ in reverse-topological order (recall that
  SCCs form an acyclic graph). So if components $\mathcal{C}_1,\mathcal{C}_2$ are
  connected by an arc from $\mathcal{C}_1$ to $\mathcal{C}_2$, then $\mathcal{C}_1$ has the larger
  number. For any program point in an SCC, let $\nn_B$ map it to the
  number assigned to this SCC.  Clearly, this numbering agrees with
  $B$; every transition rule $g$ such that $\Delta\nn_B(g)>0$ is not in
  ${\GG}_B$.  In a similar way we define $\nn_S(g)$.  Now, every cycle
  through $g$ includes an arc of $B$: this means that the end-points
  of $g$ are not in the same SCC of $\GG_B$.  Either $g$ itself is in
  $B$, or $g$ connects different SCCs; in either case, $\Delta\nn_B(g)
  \ne 0$.  Similarly, $\Delta\nn_S(g) \ne 0$.  The required conclusion
  is satisfied.
\end{proof}

\section{Summary of Experiments}
\label{app:experiments}
We provide here more information on the experimental results in Sect.\ \ref{sec:exp}.
For further details we refer to
\url{http://aprove.informatik.rwth-aachen.de/eval/MCNP}.

\subsection*{Benchmark Suite 1}

Table~\ref{t:suite1} gives the number of \emph{proofs}, the
\emph{average runtime}, and the \emph{maximum runtime} for our MCNP
implementation on the 81 examples from Suite 1.  Out of 81 MCSs of the
MC transition system, MCNP could show termination for 66 of
them.  The maximum runtime of 5.15 seconds was needed on the instance
\texttt{WTC/sipma91} consisting of 15 MC transition rules with up to
12 argument positions (source + target) and up to 60 individual order constraints in a
single monotonicity constraint.

\begin{table}[h]
\caption{Result Summary for Suite 1}
\label{t:suite1}
\begin{tabular}{cccc}
\hline\hline
Tool & Proofs & Avg.\ Time & Max.\ Time\\\hline
MCNP & 66/81   & 0.55 s & 5.15 s \\
\hline\hline
\end{tabular}
\end{table}

32 of the examples from Suite 1 originate from the
evaluation of the paper \cite{RTA09} with the
termination prover \aprove.
Table~\ref{t:suite1rta09} compares the results from
our experiments with MCNP to the experiments with \aprove.
Here the new column \emph{T/o (60 s)} denotes the number of
\emph{timeouts}, i.e.,
examples where the runs were aborted after exceeding a time
limit (here 60 seconds). The column \emph{Solved-only}
gives the number of examples that were solved by the tool in question,
but not by the other one (i.e., there was 1 example that was solved by
\aprove, but not by MCNP). Since in some of the runs timeouts
occurred, we mention two numbers for the average runtime:
\emph{Avg.\ Time (excl.\ t/o)} gives
the average runtime on the examples where the tool in question had no
timeouts, and \emph{Avg.\ Time (incl.\ t/o)} denotes the average
runtime on all examples in the example suite, where timeouts are
counted by the value of the time limit (i.e., here 60 seconds).

\begin{table}[h]
\caption{Result Summary for Suite 1
 on Instances from \protect \cite{RTA09}}
\label{t:suite1rta09}
\begin{tabular}{ccccccc}
\hline\hline
Tool & Proofs & Avg.\ Time & Avg.\ Time & Max.\ Time & T/o & Solved-only \\
 & & (excl.\ t/o) & (incl.\ t/o) & & (60 s) & \\
\hline
MCNP    & 27/32 & 0.22 s & \phantom{0}0.22 s & \phantom{$>$ 6}4.22 s & -- & 0 \\
\aprove & 28/32 & 5.30 s & 12.14 s           & $>$ 60.00 s           & 4  & 1 \\
\hline\hline
\end{tabular}
\end{table}

\subsection*{Benchmark Suite 2}

Table~\ref{t:suite2} compares the results of our experiments to those of
\costa when applied with a timeout of 120 seconds on the examples of Suite 2.  
The columns in this table
are the same as explained for Table~\ref{t:suite1rta09}.  From the 392
SCCs in the MC transition systems in this suite, MCNP could show
termination of 296 of them.  The maximum runtime for MCNP (16.31
seconds) was needed on the example \verb!Julia_10_Recursive/Test6!,
consisting of 36 MC transition rules with up to
16 argument positions and up to 51 individual order
constraints in a single monotonicity constraint.

\begin{table}[h]
\caption{Result Summary for Suite 2}
\label{t:suite2}
\begin{tabular}{ccccccc}
\hline\hline
Tool & Proofs & Avg.\ Time  & Avg.\ Time  & Max.\ Time & T/o & Solved-only \\
     &        & (excl. t/o) & (incl. t/o) &            & (120 s) & \\
\hline
MCNP    &  92/165 & 0.662 s & 0.662 s  & \phantom{$>$ 1}16.31 s & -- & \phantom{1}5 \\
\costa  & 102/165 & 0.076 s & 3.709 s  & $>$ 120.00 s           & 5  & 15 \\
\hline\hline
\end{tabular}
\end{table}

\subsection*{Benchmark Suite 3}

Table \ref{t:suite3fixabs} compares the results of our MCNP
implementation to those of a variant of \aprove\ where we 
fix the abstraction to be the same as in the
preprocessor for MCNP. Table \ref{t:suite3full} 
compares the results of MCNP to those of \aprove\ without fixing the
abstraction. The columns in these tables
are the same as explained for Table~\ref{t:suite1rta09}. 
The timeouts of MCNP on this suite may be due to
the increased complexity of the corresponding instances.
For example, \verb!Julia_10_Iterative/Infix2Postfix! consists of 319 MC
transition rules with up to 11 argument positions and up to 29 individual order
constraints in a single 
monotonicity constraint, and the example \verb!Julia_10_Iterative/Test9! has 56 MC
transition rules with up to 14 argument positions and up to 158 individual order
constraints in a single 
monotonicity constraint.

\begin{table}[h]
\caption{Result Summary for Suite 3 using \aprove\ with Fixed Abstraction}
\label{t:suite3fixabs}
\begin{tabular}{ccccccc}
\hline\hline
Tool & Proofs & Avg.\ Time  & Avg.\ Time  & Max.\ Time & T/o & Solved-only \\
     &        & (excl. t/o) & (incl. t/o) &            & (60 s) & \\
\hline
MCNP         &  63/127 & \phantom{1}2.12 s & \phantom{1}5.76 s & $>$ 60 s & 8 & \phantom{1}5 \\
\aprove\ fix &  73/127 & 11.08 s           & 14.16 s           & $>$ 60 s & 8 & 15 \\
\hline\hline
\end{tabular}
\end{table}

\begin{table}[h]
\caption{Result Summary for Suite 3 using Full \aprove}
\label{t:suite3full}
\begin{tabular}{ccccccc}
\hline\hline
Tool & Proofs & Avg.\ Time  & Avg.\ Time  & Max.\ Time & T/o & Solved-only \\
     &        & (excl. t/o) & (incl. t/o) &            & (60 s) & \\
\hline
MCNP             &  63/127 & 2.12 s & \phantom{1}5.76 s & $>$ 60 s & \phantom{1}8 &   \phantom{3}3 \\
\aprove          &  95/127 & 9.58 s & 17.12 s           & $>$ 60 s & 19 &  35 \\
\hline\hline
\end{tabular}
\end{table}

When executing MCNP with no timeout, one could show termination of 64
examples with MCNP (the proof for the additional example
\verb!Julia_10_Iterative/Test9! needs 190.6 seconds), and MCNP can
show termination of 74 of the 181 SCCs in the MCSs of this suite. 
MCNP's highest runtime is
obtained on the example \verb!Aprove_09/SortCount! with 971.7 seconds,
and it is worth noting that this example consists of 50 MC transition
rules with up to 212 individual order constraints in a single
monotonicity constraint.

\end{document}